 \documentclass[twoside]{article}
\usepackage[accepted]{aistats2018}
\usepackage{bbm}

\usepackage[utf8]{inputenc} % allow utf-8 input
\usepackage[T1]{fontenc}    % use 8-bit T1 fonts
\usepackage{hyperref}       % hyperlinks
\usepackage{url}            % simple URL typesetting
\usepackage{booktabs}       % professional-quality tables
\usepackage{amsfonts}       % blackboard math symbols
\usepackage{nicefrac}       % compact symbols for 1/2, etc.
\usepackage{microtype}      % microtypography

\usepackage{xparse}

\usepackage[ruled,vlined,linesnumbered]{algorithm2e}

\usepackage{amsmath}
\usepackage{amstext}
\usepackage{amssymb}
\usepackage{graphicx}
\usepackage{hyphenat}
 \usepackage{color}
 \usepackage{cite}
 
\usepackage[T1]{fontenc}
\usepackage[utf8]{inputenc}
\usepackage{authblk}

\NewDocumentCommand\col{g}{%
  \IfNoValueTF{#1}{\ensuremath{\mathrm{vec}}}{\ensuremath{\mathrm{vec}}\of{#1}}%
}

\NewDocumentCommand\of{og}{%
  \IfNoValueTF{#1}%
    { \IfNoValueTF{#2}{}{\!\({#2}\)} }%
    { \IfNoValueTF{#2}{\!\[{#1}\]}{\!\{{#2}\}} }%
}

\DeclareMathOperator{\Diff}{\ipaclap{D}{\raisebox{.204em}{\textpalhook}\kern.44em}\kern-.1em}
\NewDocumentCommand\diff{g}{%
  \IfNoValueTF{#1}
  {\text{\texthtd}}
  {\text{\texthtd}\of{#1}}%
}

\RenewDocumentCommand\ln{g}{%
  \IfNoValueTF{#1}{\mathrm{ln\ }}{\mathrm{ln}\of{#1}}%
}

\NewDocumentCommand\Real{og}{%
  \IfNoValueTF{#1}%
    { \IfNoValueTF{#2}{\mathcal{R}\!\!\mathpzc{e}}{\mathcal{R}\!\!\mathpzc{e}\!\{{#2}\}} }%
    { \IfNoValueTF{#2}{\mathcal{R}\!\!\mathpzc{e}\!\[{#1}\]}{\mathcal{R}\!\!\mathpzc{e}\!\({#2}\)} }%
}
\NewDocumentCommand\Imag{og}{%
  \IfNoValueTF{#1}%
    { \IfNoValueTF{#2}{\mathcal{I}\!\!\mathpzc{m}}{\mathcal{I}\!\!\mathpzc{m}\!\{{#2}\}} }%
    { \IfNoValueTF{#2}{\mathcal{I}\!\!\mathpzc{m}\!\[{#1}\]}{\mathcal{I}\!\!\mathpzc{m}\!\({#2}\)} }%
}

\RenewDocumentCommand\cos{g}{%
  \IfNoValueTF{#1}{\mathrm{cos}}{\mathrm{cos}\of{#1}}%
}
\RenewDocumentCommand\sin{g}{%
  \IfNoValueTF{#1}{\mathrm{sin}}{\mathrm{sin}\of{#1}}%
}
\RenewDocumentCommand\tan{g}{%
  \IfNoValueTF{#1}{\mathrm{tan}}{\mathrm{tan}\of{#1}}%
}
\RenewDocumentCommand\arccos{g}{%
  \IfNoValueTF{#1}{\mathrm{arccos}}{\mathrm{arccos}\of{#1}}%
}
\RenewDocumentCommand\arcsin{g}{%
  \IfNoValueTF{#1}{\mathrm{arcsin}}{\mathrm{arcsin}\of{#1}}%
}
\RenewDocumentCommand\arctan{g}{%
  \IfNoValueTF{#1}{\mathrm{arctan}}{\mathrm{arctan}\of{#1}}%
}
\RenewDocumentCommand\cot{g}{%
  \IfNoValueTF{#1}{\mathrm{cot}}{\mathrm{cot}\of{#1}}%
}

\newcommand{\floor}[1]{\ensuremath{\left\lfloor #1 \right\rfloor}}

\NewDocumentCommand\tr{g}{%
  \IfNoValueTF{#1}{\mathrm{tr}}{\mathrm{tr}\of{#1}}%
}

\NewDocumentCommand\diag{og}{%
  \IfNoValueTF{#1}%
    { \IfNoValueTF{#2}{\ensuremath{\mathrm{diag}}}{\ensuremath{\mathrm{diag}\of{#2}}} }%
    { \IfNoValueTF{#2}{\ensuremath{\mathrm{diag}\of[#1]}}{\ensuremath{\mathrm{diag}\of[]{#2}}} }%
}

\RenewDocumentCommand\exp{g}{%
  \IfNoValueTF{#1}{\ensuremath{\mathrm{exp}}}{\ensuremath{\mathrm{exp}}\of{#1}}%
}

\newcommand{\E}[2][]{\Operator[#1]{E}{#2}}

\NewDocumentCommand\C{g}{%
  \IfNoValueTF{#1}{\mathrm{Cov}}{\mathrm{Cov}\of{#1}}%
}
\renewcommand{\(}{\ensuremath{\left(}}
\renewcommand{\)}{\ensuremath{\right)}}
\renewcommand{\[}{\ensuremath{\left[}}
\renewcommand{\]}{\ensuremath{\right]}}

\let\oldBracketLeft\{
\let\oldBracketRight\}
\renewcommand{\{}{\ensuremath{\left\oldBracketLeft}}
\renewcommand{\}}{\ensuremath{\right\oldBracketRight}}

\NewDocumentCommand\F{og}{%
  \IfNoValueTF{#1}%
    { \IfNoValueTF{#2}{\mathcal{F}}{\mathcal{F}\!\{{#2}\}} }%
    { \IfNoValueTF{#2}{\mathcal{F}\!\[{#1}\]}{\mathcal{F}\!\({#2}\)} }%
}
\NewDocumentCommand\FInv{og}{%
  \IfNoValueTF{#1}%
    { \IfNoValueTF{#2}{\mathcal{F}^{-1}}{\mathcal{F}^{-1}\!\{{#2}\}} }%
    { \IfNoValueTF{#2}{\mathcal{F}^{-1}\!\[{#1}\]}{\mathcal{F}^{-1}\!\({#2}\)} }%
}

\NewDocumentCommand\rect{g}{%
  \IfNoValueTF{#1}
  {\ensuremath{\mathrm{rect}}}
  {\ensuremath{\mathrm{rect}\of{#1}}}%
}

\NewDocumentCommand\sinc{g}{%
  \IfNoValueTF{#1}
  {\ensuremath{\mathrm{sinc}}}
  {\ensuremath{\mathrm{sinc}\of{#1}}}%
}

\NewDocumentCommand\supp{g}{%
  \IfNoValueTF{#1}
  {\ensuremath{\mathrm{supp}}}
  {\ensuremath{\mathrm{supp}\of{#1}}}%
}

\let\oldMathcal\mathcal
\renewcommand{\mathcal}[1]{\ensuremath{\oldMathcal{#1}}}

\makeatletter
\def\foreach#1#2#3{%
  \@test@foreach{#1}{#2}#3,\@end@token
}

\def\@swallow#1{}

\def\@test@foreach#1#2{%
  \@ifnextchar\@end@token%
    {\@swallow}%
    {\@foreach{#1}{#2}}%
}

\def\@foreach#1#2#3,#4\@end@token{%
  #1{#2}{#3}%
  \@test@foreach{#1}{#2}#4\@end@token%
}
\makeatother

\newtheorem{theorem}{Theorem}[section]
\newtheorem{lemma}[theorem]{Lemma}

\newtheorem{remark}{Remark}

\newenvironment{proof}[1][Proof]{\begin{trivlist}
\item[\hskip \labelsep {\bfseries #1}]}{\end{trivlist}}
\newenvironment{definition}[1][Definition]{\begin{trivlist}
\item[\hskip \labelsep {\bfseries #1}]}{\end{trivlist}}

\def\l|{\left|}
\def\r|{\right|}
\def\l({\left(}
\def\r){\right)}
\def\l[{\left[}
\def\r]{\right]}

\renewcommand{\E}[1]{\mathbb{E}\left[ #1 \right]}
\newcommand{\half}{{\frac{1}{2}}}

\begin{document} 
\title{Scalable Hash-Based Estimation of Divergence Measures}

\author{Morteza Noshad}

\author{Alfred O. Hero III}
\affil{University of Michigan,
Electrical Engineering and Computer Science,
Ann Arbor, Michigan, U.S.A}

\renewcommand\Authands{ and }

\date{}
\maketitle

\begin{abstract}
We propose a scalable divergence estimation method based on hashing. Consider two continuous random variables $X$ and $Y$ whose densities have bounded support. We consider a particular locality sensitive random hashing, and consider the ratio of samples in each hash bin having non-zero numbers of Y samples. We prove that the weighted average of these ratios over all of the hash bins converges to f-divergences between the two samples sets. We show that the proposed estimator is optimal in terms of both MSE rate and computational complexity. We derive the MSE rates for two families of smooth functions; the H\"{o}lder smoothness class and differentiable functions. In particular, it is proved that if the density functions have bounded derivatives up to the order $d/2$, where $d$ is the dimension of samples, the optimal parametric MSE rate of $O(1/N)$ can be achieved.  The computational complexity is shown to be $O(N)$, which is optimal. To the best of our knowledge, this is the first empirical divergence estimator that has optimal computational complexity and achieves the optimal parametric MSE estimation rate.

\end{abstract}
\section{Introduction}

Information theoretic measures such as Shannon entropy, mutual information, and the Kullback-Leibler (KL) divergence have a broad range of applications in information theory, statistics and machine learning \cite{cover2012,moon2014,structure2016}.
When we have two or more data sets and we are interested in finding the correlation or dissimilarity between them, Shannon mutual information or KL-divergence is often used.
R\'{e}nyi and f-divergence measures are two well studied generalizations of KL-divergence which comprise many important divergence measures such as KL-divergence, total variation distance, and $\alpha$-divergence \cite{rrnyi1961,ali}.  

Non\hyp parametric estimators are a major class of divergence estimators, for which minimal assumptions on the density functions are considered. Some of the non-parametric divergence estimators are based on density plug-in estimators such as $k$-NN \cite{poczos2011}, KDE \cite{Poczos2014_2}, and histogram \cite{wang2009}. A few researchers, on the other hand, have proposed direct estimation methods such as graph theoretic nearest neighbor ratio (NNR) \cite{Noshad2017}. In general, plug-in estimation methods suffer from high computational complexity, which make them unsuitable for large scale applications.

Recent advances on non-parametric divergence estimation have been focused on the MSE convergence rates of the estimator. 
Singh et al in \cite{Poczos2014_2} proposed a plug-in KDE estimator for R\'{e}nyi divergence that achieves the MSE rate of $O(1/N)$ when the densities are at least $d$ times differentiable, and the support boundaries are sufficiently smooth. Kandasamy et al proposed a similar plug-in KDE estimator and extend the optimal MSE rate to densities that are at least $d/2$ differentiable \cite{kandasamy}. However, they ignore a major source of error due to the boundaries. 
Moon et al proposed a weighted ensemble method to improve the MSE rate of plug-in KDE estimators \cite{Kevin16}. The proposed estimator for f-divergence achieves the optimal MSE rate when the densities are at least $(d+1)/2$ times differentiable. They also assume stringent smoothness conditions at the support set boundary. 

Noshad et al proposed a graph theoretic direct estimation method based on nearest neighbor ratios (NNR) \cite{Noshad2017}. Their estimator is simple and computationally more tractable than other competing estimators, and can achieve the optimal MSE rate of $O(1/N)$ for densities that are at least $d$ times differentiable. Although their basic estimator does not require any smoothness assumptions on the support set boundary, the ensemble estimator variant of their estimator does. 

In spite of achieving the optimal theoretical MSE rate by aforementioned estimators, there remain serious. The first challenge is the high computational complexity of the estimator. Most KDE based estimators require runtime complexity of $O(N^2)$, which is not suitable for large scale applications. The NNR estimator proposed in \cite{Noshad2017} has the runtime complexity of $O(kN\log N)$, which is faster than the previous estimators. However, in \cite{Noshad2017} they require $k$ to grow sub-linearly with $N$, which results in much higher complexity than linear runtime complexity. The other issue is the smoothness assumptions made on the support set boundary. Almost all previously proposed estimators assume extra smoothness conditions on the boundaries, which may not hold practical applications.  
For example, the method proposed in \cite{Poczos2014_2} assumes that the density derivatives up to order $d$ vanish at the boundary. Also it requires numerous computations at the support boundary, which become complicated when the dimension increases. The Ensemble NNR estimator in \cite{Noshad2017} assumes that the density derivatives vanish at the boundary.
To circumvent this issue, Moon et al \cite{Kevin16} assumed smoothness conditions at the support set boundary. However, these conditions may not hold in practice.

In this paper we propose a low complexity divergence estimator that can achieve the optimal MSE rate of $O(1/N)$ for the densities with bounded derivatives of up to $d/2$. Our estimator has optimal runtime complexity of $O(N)$, which makes it an appropriate tool for large scale applications. Also in contrast to other competing estimators, our  estimator does not require stringent smoothness assumptions on the support set boundary. 

The structure of the proposed estimator borrows ideas from hash based methods for KNN search and graph constructions problems \cite{hash_KNN_graph, LSH_KNN}, as well as from the NNR estimator proposed in \cite{Noshad2017}. The advantage of hash based methods is that they can be used to find the approximate nearest neighbor points with lower complexity as compared to the exact $k$-NN search methods. This suggests that fast and accurate algorithms for divergence estimation may be derived from hashing approximations of k-NN search.
Noshad et al \cite{Noshad2017} consider the $k$-NN graph of Y in the joint data set $(X,Y)$, and show that the average exponentiated ratio of the number of X points to the number of Y points among all $k$-NN points is proportional to the  R\'{e}nyi divergence between the X and Y densities. It turns out that for estimation of the density ratio around each point we really do not need to find the exact $k$-NN points, but only need sufficient local samples from X and Y around each point. By using a randomized locality sensitive hashing (LSH), we find the closest points in Euclidean space. In this manner, applying ideas from the NNR estimation and hashing techniques to KNN search problem, we obtain a more efficient divergence estimator.
Consider two sample sets $X$ and $Y$ with a bounded density support. We use a particular two-level locality sensitive random hashing, and consider the ratio of samples in each bin with a  number of Y samples. We prove that the weighted average of these ratios over all of the bins can be made to converge almost surely to f-divergences between the two samples populations.
We also argue that using the ensemble estimation technique provided in \cite{moon2014}, we can achieve the optimal parametric rate of $O(1/N)$. Furthermore, using a simple algorithm for  online estimation method has $O(N)$ complexity and $O(1/N)$ convergence rate, which is the first optimal online estimator of its type. 

%\input{arxiv_introduction}
%%%%%%%%%%%%    Definition of Renyi Divergence

\section{Hash-Based Estimation}
In this section, we first introduce  the f-divergence measure and propose a hash-based estimator. We outline the main theoretical results which will be proven in section \ref{Proof_Section}.

Consider two density functions $f_1$ and $f_2$ with common bounded support set $\mathcal{X}\subseteq \mathbb{R}^d$.

The f-divergence is defined as follows \cite{ali}.

\begin{align}\label{def_f_div}
D_g\left(f_1(x)\vert\vert f_2(x)\right) & :=\int g\of{\frac{f_1(x)}{f_2(x)}}f_2(x)dx \nonumber\\
& =\mathbb{E}_{f_2}\of[g\of{\frac{f_1(x)}{f_2(x)}}],
\end{align}
where $g$ is a smooth and convex function such that $g(1)=0$. KL-divergence, Hellinger distance and total variation distance are particular cases of this family. Note that for estimation, we don’t need convexity of $g$ and $g(1)=0$. conditions. Assume that the densities are lower bounded by $C_L>0$ and upper bounded by $C_U$. Assume $f_1$ and $f_2$ belong to the H\"{o}lder smoothness class with parameter $\gamma$:

%%% Holder Smoothness %%%%

\begin{definition}\label{Holder}
Given a support $\mathcal{X} \subseteq \mathbb{R}^d$, a function $f:\mathcal{X} \to \mathbb{R}$ is called H\"{o}lder continuous with parameter $0<\gamma\leq 1$, if there exists a positive constant $G_f$, possibly depending on $f$, such that 
\begin{equation}\label{Holder_eq}
|f(y)-f(x)|\leq G_f\|y-x\|^{\gamma},
\end{equation}
for every $x\neq y \in \mathcal{X}$.
\end{definition}

%%%%  Lipschitz Definition %%%
The function $g$ in \eqref{def_f_div} is also assumed to be Lipschitz continuous; i.e. $g$ is H\"{o}lder continuous with $\gamma=1$.
%The function $g(x)$ defined in \eqref{def_f_div} is also %assumed to be Lipschitz continuous, defined as follows

%\begin{definition}
%Given a support $\mathcal{X} \subseteq \mathbb{R}^d$, a %function $f:\mathcal{X} \to \mathbb{R}$ is called Lipschitz %continuous if there exists a constant $H_f>0$ such that 
%\begin{equation}
%|f(y)-f(x)|\leq H_f\|y-x\|,
%\end{equation}
%for every $x\neq y \in \mathcal{X}$.
%\end{definition}

\begin{remark}
The $\gamma$-H\"{o}lder smoothness family comprises a large class of continuous functions including continuously differentiable functions and Lipschitz continuous functions.
Also note that for $\gamma > 1$, any $\gamma$–H\"{o}lder continuous function on any bounded and continuous support is constant.  
\end{remark}

\begin{definition}
[Hash-Based Divergence Estimator: ] Consider the i.i.d samples $X=\{X_1,...,X_N\}$ drawn from $f_1$ and $Y=\{Y_1,...,Y_M\}$ drawn from $f_2$. Define the fraction $\eta:=M/N$. We define the set $Z:=X\cup Y$.  
We define a positive real valued constant $\epsilon$ as a user-selectable parameter of the estimator to be defined in \ref{est_f_def}. We define the hash function $H_1: \mathbb{R}^d\to \mathbb{Z}^d$ as 
\begin{align}\label{H1_def}
H_1(x)=\of[h_1(x_1), h_1(x_2), ... ,h_1(x_d)], 
\end{align}

where $x_i$ is the projection of $x$ on the $i$th coordinate, and $h_1(x): \mathbb{R}\to \mathbb{Z}$ is defined as

\begin{align}
h_1(x)=\floor{\frac{x+b}{\epsilon}},
\end{align}
for fixed $b$. Let $\mathcal{F}:=\{1,2,..,F\}$, where $F:=c_HN$ and $c_H$ is a fixed real number. We define a random hash function $H_2:\mathbb{Z}^d\to \mathcal{F}$ with a uniform density on the output and consider the combined hashing $H(x):=H_2(H_1(x))$, which maps the points in $\mathbb{R}^d$ to $\mathcal{F}$.

Consider the mappings of the sets $X$ and $Y$ using the hash function $H(x)$, and define the vectors $\mathcal{N}$ and  $\mathcal{M}$ to respectively contain the number of collisions for each output bucket from the set $\mathcal{F}$. We represent the bins of the vectors $\mathcal{N}$ and  $\mathcal{M}$ respectively by $N_i$ and $M_i$, $1\leq i \leq F$. 

The hash based f-divergence estimator is defined as
\begin{equation}\label{est_f_def}
\widehat{D}_{g}(X,Y):=\max\{\frac{1}{M}\sum_{\substack{i\leq F \\ M_i>0}} M_i\widetilde{g}\of{\frac{\eta N_i}{M_i}},0\},
\end{equation}
where $\widetilde{g}(x):=\max\{g(x),g\of{C_L/C_U}\}$. 

\end{definition}

Note that if the densities $f_1$ and $f_2$ are almost equal, then for each point $Y_i$, $N_i \approx M_i$, and thus $\widehat{D}_{\alpha}(X,Y)$ and $\widehat{D}_{g}(X,Y)$ tend to zero, as required. 
In the following theorems we state upper bounds on the bias and variance rates.
Let $\mathbb{B}[\hat{T}]=\mathbb{E}[\hat{T}]-T$ and $\mathbb{V}[\hat{T}]=\mathbb{E}[\hat{T}^2]-\mathbb{E}[\hat{T}]^2$, respectively represent the bias and variance of $\hat{T}$, which is an estimator of the parameter $T$. Then, the following provides a bound on the bias of the proposed estimator.

\begin{theorem} \label{bias_theorem}
Assume that $f_1$ and $f_2$ are density functions with bounded common support set $\mathcal{X}\in\mathbb{R}^d$ and satisfying $\gamma$-H\"{o}lder smoothness.
 The bias of the proposed estimator for f-divergence with function $g$ can be bounded as
\begin{align} \label{bias_Renyi}
\mathbb{B}\of[\widehat{D}_{g}(X,Y)]= O\of{\epsilon^\gamma}+O\of{\frac{1}{N\epsilon^d}},\nonumber
\end{align}
\end{theorem}
where $c_1$ is a positive real constant. 

\begin{remark}
In order for the estimator to be asymptotically unbiased, $\epsilon$ needs to be a function of $N$. The optimum bias rate of $O\of{\of{\frac{1}{N}}^{\gamma/(\gamma+d)}}$ can be achieved for $\epsilon=\of{\frac{1}{N}}^{\gamma/(\gamma+d)}$. 
\end{remark}

\begin{theorem}\label{variance}
Let $\eta=M/N$ be fixed. The variance of the  estimator \ref{est_f_def} can be bounded as 
\begin{align}
\mathbb{V}\of[\widehat{D}_{g}(X,Y)]\leq O\of{\frac{1}{N}}.
\end{align}
\end{theorem}

\begin{remark}\label{variance_others}
The same variance bound holds for the random variable $\rho_i:=\frac{N_i}{M_i}$. The bias and variance results easily extend to R\'{e}nyi divergence estimation.
\end{remark}

%%%%%%%%%%%%%%%%%%%%%%%%%%%%%

\begin{algorithm} \label{algo}
\DontPrintSemicolon
\SetKwInOut{Input}{Input}\SetKwInOut{Output}{Output}
% Input:
\Input{Data sets $X=\{X_1,...,X_N\}$, $Y=\{Y_1,...,Y_M\}$}

\BlankLine
\tcc{Find the sets of all hashed points in X and Y}
$X'\leftarrow H(X)$.\\
$Y'\leftarrow H(Y)$.\\

\For {each $i\in \mathcal{F}$}{
		 \tcc{Find the number of collisions at bin i}
		 $N_i\leftarrow |X'=i|$\\
         $M_i\leftarrow |Y'=i|$
		 }
$\widehat{D} \leftarrow \max\{\frac{1}{M}\sum_{\substack{M_i>0}} M_i\widetilde{g}\of{\eta N_i/M_i},0\},$

\Output{$\widehat{D}$}

\caption{Histogram Estimator of f-Divergence }
\end{algorithm}

%%%% Ensemble Estimator %%%%%
We next show that, when $f_1$ and $f_2$ belong to the family of differentiable densities, we can improve the bias rate by applying the ensemble estimation approach in \cite{Kevin16,structure2016}. The Ensemble Hash-based (EHB) estimator is defined as follows.

\begin{definition}
Assume that the density functions are in the H\"{o}lder space $\Sigma(\gamma,L)$, which consists of functions on $\mathcal{X}$ continuous derivatives up to order $q=\floor{\gamma}\geq d$  and the $q$th partial derivatives are H\"{o}lder continuous with exponent $\gamma'=:\gamma-q$. 
Let $\mathcal{T}:=\{t_1,...,t_T\}$ be a set of index values with $t_i<c$, where $c>0$ is a constant. Let $\epsilon(t):=\floor{tN^{-1/2d}}$. The weighted ensemble estimator is defined as 

\begin{align}\label{EHB_def}
\widehat{D}_w:=\sum_{t\in \mathcal{T}}w(t)\widehat{D}_{\epsilon(t)},
\end{align}
where $\widehat{D}_{\epsilon(t)}$ is the hash based estimator of f-divergence, with the hashing parameter of $\epsilon(t)$. 
\end{definition}

\begin{theorem} \label{ensemble_theorem}
Let $T>d$ and  $w_0$ be the solution to:
\begin{align}
\min_w &\qquad \|w\|_2 \nonumber\\
\textit{subject to} &\qquad \sum_{t\in \mathcal{T}}w(t)=1, \nonumber\\
&\qquad \sum_{t\in \mathcal{T}}w(t)t^{i/d}=0, i\in \mathbb{N}, i\leq d.
\end{align}
Then the MSE rate of the ensemble estimator $\widehat{D}_{w_0}$ is $O(1/N)$.
\end{theorem}

\section{Online Divergence Estimation}

In this section we study the problem of online divergence estimation. In this setting we consider two data steams $X=\{X_{1},X_{2},...,X_{N}\}$ and $Y=\{Y_{1},Y_{2},...,Y_{N}\}$ with i.i.d samples, and we are interested in estimating the divergence between two data sets. The number of samples increase over time and an efficient update of the divergence estimate is desired. The time complexity of a batch update, which uses the entire update batch to compute the estimate at each time point, is $O(N)$, and it may not be so effective in cases which we need quick detection of any change in the divergence function. 

Algorithm \ref{Online} updates the divergence with amortized runtime complexity of order $O(1)$. Define the sets $X^N:=\{X_i\}_{i=1}^N$, $Y^N:=\{Y_i\}_{i=1}^N$, the number of $X$ and $Y$ samples in each partition, and the divergence estimate between $X^N$ and $Y^N$. Consider updating the estimator with new samples $X_{N+1}$ and $Y_{N+1}$. In the first and second lines of algorithm \ref{Online}, the new samples are added to the datasets and the values of  $N_i$ and $M_i$ of the bins in which the new samples fall. We can find these bins in $O(1)$ using a simple hashing. Note that once $N_i$ and $M_i$ are updated, the divergence measure can be updated, but the number of bins is not increased, by Theorem \ref{bias_theorem}, it is clear that the bias will not be reduced. Since increasing the number of bins requires recomputing the bin partitions, a brute force rebinning approach would have order $O(N)$ complexity, and it were updated $N$ times, the total complexity would be $O(N^2)$. Here we use a trick and update the hash function only when $N+1$ is a power of $2$. In the following theorem, which is proved in appendix, we show that the MSE rate of this algorithm is order $O(1/N)$ and the total rebinngn computational complexity is order $O(N)$.

\begin{theorem}\label{Online_Theorem}
MSE rate of the online divergence estimator shown in Algorithm \ref{algo} is order $O(1/N)$ and the total computational complexity is order $O(N)$.
\end{theorem}

%%%%%%%%%%%%%%%%%%%%%%%%%%%%%
\begin{algorithm} \label{algo}
\DontPrintSemicolon
\SetKwInOut{Input}{Input}\SetKwInOut{Output}{Output}
% Input:
\Input{$X^N:=\{X_i\}_{i=1}^N,Y^N:=\{Y_i\}_{i=1}^N$\\
$\widehat{D} =\widehat{D}\of{X^N,Y^N}$\\
$(N_i, M_i)$\\
$(X_{N+1}, Y_{N+1})$}

\BlankLine
Add $X_{N+1}$ and Update $N_k$ s.t $H(X_{N+1})=k$. \;
Add $Y_{N+1}$ and Update $M_l$ s.t $H(Y_{N+1})=l$. \;
If $N+1=2^{i}$ for some $i$, Then \;
\qquad Update $\epsilon$ to the optimum value\;
\qquad Re-hash X and Y\;
\qquad Recompute $N_i$ and $M_i$ for $0\leq i \leq F$\;

Update $\widehat{D}$

\Output{$\widehat{D}$}

\caption{Online Divergence Estimation}
\label{Online}
\end{algorithm}

%%%%%%%%%%%%%%%%%%%%%%
%\input{known_support}
\section{Proofs}\label{Proof_Section}

In this section we derive the bias bound for the densities in H\"{o}lder smoothness class, stated in Theorem \ref{bias_theorem}. For the proofs of variance bound in Theorem \ref{variance}, convergence rate of EHB estimator in Theorem \ref{ensemble_theorem}, and online divergence estimator in Theorem \ref{Online_Theorem}, we refer the reader to the Appendix, provided as a supplementary pdf file.

Consider the mapping of the $X$ and $Y$ points by the hash function $H_1$, and let the vectors $\{V_i\}_{i=1}^{L}$ represent the distinct mappings of $X$ and $Y$ points under $H_1$. Here $L$ is the number of distinct outputs of $H_1$. In the following lemma we prove an upper bound on $L$.

\begin{lemma}\label{Lemma_bound_L}
Let $f(x)$ be a density function with bounded support $\mathbb{X}\subseteq\mathbb{R}^d$. 
Then if $L$ denotes the number of distinct outputs of the hash function $H_1$ (defined in \eqref{H1_def}) of i.i.d points with density $f(x)$, we have

\begin{align}
L &\leq O\of{\frac{1}{\epsilon^d}}.
\end{align}

\end{lemma}

\begin{proof}

Let $x=\[x_1,x_2,...,x_d\]$ and define $\mathcal{X}_I$ as the region defined as 

\begin{align}
\mathbb{X}_I:=\{x|-c_X\leq x_i \leq c_X , 1\leq i \leq d \},
\end{align}
where $c_X$ is a constant such that  $\mathbb{X}\subseteq \mathbb{X}_I$.

$L$ is clearly not greater than the total number of bins created by splitting the region $\mathbb{X}$ into partitions of volume $\epsilon^d$. So we have
\begin{align}
L\leq \frac{(2c_X)^d}{\epsilon^d}.
\end{align}

\end{proof}
%%%%%%%%%%%%%%%%%%%%%%%%%%

\textbf{Proof of Theorem \ref{bias_theorem}}
Let $\{N'_i\}_{i=1}^{L}$ and $\{M'_j\}_{j=1}^{L}$ respectively denote the number of collisions of $X$ and $Y$ points in the bins $i$ and $j$, using the hash function $H_1$. 
$E_i$ stands for the event that there is no collision in bin $i$ for the hash function $H_2$ with inputs  $\{V_i\}_{i=1}^{L}$. We have

\begin{align}\label{PEi}
P(E_i)&=\of{1-\frac{1}{F}}^L+L\of{\frac{1}{F}}\of{\frac{F-1}{F}}^{L-1}\nonumber\\
&=1-O\of{\frac{L}{F}}.
\end{align}
By definition, $$\widehat{D}_{g}(X,Y):= \frac{1}{M}\sum_{\substack{i\leq F \\ M_i>0}} M_i\widetilde{g}\of{\frac{\eta N_i}{M_i}}.$$

Therefore
\begin{align}\label{J2sum_1}
\E{\widehat{D}_{g}(X,Y)} &=\frac{1}{M}\E{\sum_{\substack{i\leq F \\ M_i>0}} M_i\widetilde{g}\of{\frac{\eta N_i}{M_i}}}\nonumber\\
%%% Second Line%%%
&=\frac{1}{M}\sum_{\substack{i\leq F \\ M_i>0}}P(E_i)\E{ M_i\widetilde{g}\of{\frac{\eta N_i}{M_i}}\middle\vert E_i} \nonumber\\
&+\frac{1}{M}\sum_{\substack{i\leq F \\ M_i>0}}P(\overline{E_i})\E{ M_i\widetilde{g}\of{\frac{\eta N_i}{M_i}}\middle\vert  \overline{E_i} }.
\end{align}

We represent the second term in \eqref{J2sum_1} by $\mathbb{B}_H$. $\mathbb{B}_H$ has the interpretation as the bias error due to collisions in hashing. Remember that  $\overline{E_i}$ is defined as the event that there is a collision at bin $i$ for the hash function $H_2$ with inputs $\{V_i\}_{i=1}^{L}$. For proving as upper bound on $\mathbb{B}_H$, we first need to compute an upper bound on $\sum_{i=1}^L\E{ M_i\middle\vert  \overline{E_i}^j}$. This is stated in the following lemma.
%%%%%%%%%%%%%%%%%%%%%%%%%

\begin{lemma}\label{E_Mi}
We have
\begin{align}
\sum_{\substack{i\leq F \\ M_i>0}}\E{ M_i\middle\vert  \overline{E_i}}\leq O\of{L}
\end{align}
\end{lemma}

\begin{proof}
Define $\mathcal{A}_i:= \{j: H_2(V_j)=i\}$. For each $i$ we can rewrite $M_i$ as
\begin{align}
M_i=\sum_{j=1}^L \mathbbm{1}_{\mathcal{A}_i}(j)M_j'.
\end{align}
Thus, 
\begin{align}\label{Collision_bias_proof_1}
\sum_{\substack{i\leq F \\ M_i>0}}\E{ M_i\middle\vert  \overline{E_i}}&=\sum_{\substack{i\leq F \\ M_i>0}}\E{ \sum_{j=1}^L\mathbbm{1}_{\mathcal{A}_i}(j)M_j'\middle\vert\overline{E_i}}\nonumber\\
% Second Line
&=\sum_{\substack{i\leq F \\ M_i>0}}\sum_{j=1}^L M_j'\E{ \mathbbm{1}_{\mathcal{A}_i}(j)\middle\vert\overline{E_i}}\nonumber\\
% Third Line
&=\sum_{\substack{i\leq F \\ M_i>0}}\sum_{j=1}^L M_j' P\of{ j\in\mathcal{A}_i|\overline{E_i}}\nonumber\\
% Fourth Line
&=\sum_{\substack{i\leq F \\ M_i>0}}\sum_{j=1}^L M_j' \frac{P\of{ j\in\mathcal{A}_i,\overline{E_i}}}{P(\overline{E_i})},
\end{align}
where $P\of{ j\in\mathcal{A}_i,\overline{E_i}}$ and $P(\overline{E_i})$ can be derived as
\begin{align}\label{P1}
P\of{ j\in\mathcal{A}_i,\overline{E_i}}&=\frac{1}{F}\of{1-\of{\frac{F-1}{F}}^{L-1}}=O\of{\frac{L}{F^2}},
\end{align}
and
\begin{align}\label{P2}
P(\overline{E_i})=1-P(E_i)=O\of{\frac{L}{F}}.
\end{align}

Plugging in \eqref{P1} and \eqref{P2} in \eqref{Collision_bias_proof_1} results in

\begin{align}\label{Collision_bias_proof_2}
\sum_{\substack{i\leq F \\ M_i>0}}\E{ M_i\middle\vert  \overline{E_i}}&=
\sum_{\substack{i\leq F \\ M_i>0}}\sum_{j=1}^L M_j' O\of{\frac{1}{F}}\nonumber\\
% Second Line
&=\sum_{\substack{i\leq F \\ M_i>0}} O\of{\frac{M}{F}}\nonumber\\
&=O\of{L},
\end{align}
where in the third line we use $\eta=M/N$ and $F=c_HN$. Now in the following lemma we prove a bound on $\mathbb{B}_H$.

\end{proof}
%%%%%%%%%%%%%%%%%%%%%%

\begin{lemma}
Let $L$ denote the number of distinct outputs of the hash function $H_1$ of the $X$ and $Y$ sample points. The bias of estimator \eqref{est_f_def} due to hashing collision can be upper bounded by
\begin{align}
\mathbb{B}_H \leq O\of{\frac{L^2}{N^2}}
\end{align}
\end{lemma}

\begin{proof}

From the definition of $\mathbb{B}_H$ we can write

\begin{align}
 \mathbb{B}_H:&=\frac{1}{M}\sum_{\substack{i\leq F \\ M_i>0}}P(\overline{E_i})\E{ M_i\widetilde{g}\of{\frac{\eta N_i}{M_i}}\middle\vert  \overline{E_i}} \nonumber\\
&= \frac{P(\overline{E_1})}{M}\sum_{\substack{i\leq F \\ M_i>0}}\E{ M_i\widetilde{g}\of{\frac{\eta N_i}{M_i}}\middle\vert \overline{E_i}}\nonumber\\
 %% Second Line
 &\leq \frac{P(\overline{E_1})\widetilde{g}(R_{max})}{M}\sum_{\substack{i\leq F \\ M_i>0}}\E{ M_i\middle\vert  \overline{E_i}}\nonumber\\
 %% Third Line
 &= \frac{P(\overline{E_1})\widetilde{g}(R_{max})}{M}O(L)\nonumber\\
  %% Third Line
 &= O\of{\frac{L^2}{N^2}},
\end{align}
where in the second line we used the fact that $P(\overline{E_i})=P(\overline{E_1})$. In the third line we used the upper bound for $\tilde{g}$, and in the fourth line we used the result in equation \eqref{Collision_bias_proof_2}.

\end{proof}

%%%%%%%%%%%%%%%%%%%%%%%
Now we are ready to continue the proof of the bias bound in \eqref{J2sum_1}. Let $E$ be defined as the event that there is no collision for the hash function $H_2$, and all of its outputs are distinct, that is, $E=\cap_{i=1}^F  E_i$

\eqref{J2sum_1} can be written as 

\begin{align}
&\E{\widehat{D}_{g}(X,Y)} \nonumber\\
&\quad=\frac{1}{M}\sum_{\substack{i\leq F \\ M_i>0}}P(E_i)\E{ M_i\widetilde{g}\of{\frac{\eta N_i}{M_i}}\middle\vert E_i}+ O\of{\frac{L}{F}}\nonumber\\
%%% Second Line %%%%
&\quad=\frac{P(E_1)}{M}\sum_{\substack{i\leq F \\ M_i>0}}\E{ M_i\widetilde{g}\of{\frac{\eta N_i}{M_i}}\middle\vert E_i}+ O\of{\frac{L}{F}}\nonumber\\
%%% Third Line %%%%
&\quad=\frac{P(E_1)}{M}\sum_{\substack{i\leq F \\ M_i>0}}\E{ M_i\widetilde{g}\of{\frac{\eta N_i}{M_i}}\middle\vert E}+ O\of{\frac{L}{F}}\label{B1_3rd}\\
%%% Forth Line %%%%
&\quad=\frac{P(E_1)}{M}\E{\sum_{\substack{i\leq F \\ M_i>0}} M_i\widetilde{g}\of{\frac{\eta N_i}{M_i}}\middle\vert E}+ O\of{\frac{L}{F}}\nonumber\\
%%% Fifth Line %%%%
&\quad=\frac{P(E_1)}{M}\E{\sum_{i=1}^L M'_i\widetilde{g}\of{\frac{\eta N'_i}{M'_i}}\middle\vert E}+ O\of{\frac{L}{F}}\label{B1_5th}\\
%%% sixth Line %%%%
&\quad=\frac{1-O(L/F)}{M}\E{ \sum_{i=1}^M\widetilde{g}\of{\frac{\eta N'_i}{M'_i}}}+O\of{\frac{L}{F}}\label{B1_6th}\\
%%% seven Line %%%%
&\quad=\mathbb{E}_{Y_1\sim f_2(x)} \E{ \widetilde{g}\of{\frac{\eta N'_1}{M'_1}}\middle\vert Y_1}+ O\of{\frac{L}{F}}\label{B1_7th},
\end{align}
where in \eqref{B1_3rd} we have used the fact that conditioned on $E_i$, $N_i$ and $M_i$ are independent of $E_j$ for $i\neq j$. In \eqref{B1_5th} since there is no collision in $H_2$, $M_i'$ and $N_i'$ are equal to $M_j$ and $N_j$ for some $i$ and $j$. Equation \eqref{B1_6th} is because the values $M_i'$ and $N_i'$ are independent of the hash function $H_2$ and its outputs, and finally in equation \eqref{B1_7th}, we used the fact that each set $N_i'$ and $M_i'$ are i.i.d random variables.

At this point, assuming that the variance of $\frac{N'_1}{M'_1}$ is upper bounded by $O(1/N)$ and using (Lemma 3.2 in \cite{Noshad2017}), we only need to derive $\E{\frac{N'_1}{M'_1}}$, and then we can simply find the RHS in \eqref{B1_7th}. Note that $N'_i$ and $M'_i$ are independent and have binomial distributions with the respective means of $NP_i^X$ and $MP_i^Y$, where $P_i^X$ and $P_i^Y$ are the probabilities of mapping $X$ and $Y$ points with the respective densities $f_0$ and $f_1$ into bin $i$. Hence,
\begin{align}\label{E_ratio}
&\E{\frac{N'_1}{M'_1}\middle\vert Y_1} = \E{N'_1\middle\vert Y_1} \E{{M'_1}^{-1}\middle\vert Y_1}.
\end{align}

Let $B_i$ denote the area for which all the points map to the same vector $V_i$. $\E{N'_i}$ can be written as:

\begin{align}\label{ENi_Expansion}
\E{N'_i} &=N\int_{x\in B_i}f_1(x)dx\nonumber\\
&=N\int_{x\in B_i}f_1(Y_i)+O(\|x-Y_i\|^\gamma)dx\nonumber\\
&=N\epsilon^df_1(Y_i)+ N\int_{x\in B_i}O(\|x-Y_i\|^\gamma)dx\nonumber\\
&= N\epsilon^df_1(Y_i)+ N\int_{x\in B_i+Y_i}O(\|x\|^\gamma)dx,
\end{align}
where in the second equality we have used the definition in \eqref{Holder_eq}.
Let define $B'_i:=\frac{1}{\epsilon}B_i+\frac{1}{\epsilon}Y_i$ and 

\begin{equation}\label{CY}
C_\gamma(Y_i):=\int_{x\in B'_i}\|x\|^\gamma dx.
\end{equation}

Note that $C_\gamma(Y_i)$ is a constant independent of $\epsilon$, since the volume of $B'_i$ is independent of $\epsilon$. By defining $x'=x/\epsilon$ we can write
\begin{align}\label{int_moment}
\int_{x'\in B'_i}\|x\|^\gamma dx&= \int_{x'\in B'_i}\epsilon^\gamma \|x\|^\gamma (\epsilon^d dx')
=C_\gamma(Y_i)\epsilon^{\gamma+d}
\end{align}

Also note that since the number of $X$ and $Y$ points in each bin are independent we have $\E{N'_i|Y_i}=\E{N'_i}$, and therefore 
\begin{align}
\E{N'_i|Y_i}=N\epsilon^df_1(Y_i)+O\of{N\epsilon^{\gamma+d}C_\gamma(Y_i)}.
\end{align}

Next, note that $\E{M'_i|Y_i}$ has a non-zero binomial distribution, for which the first order inverse moment can be written as \cite{znidaric2005}:

\begin{align}\label{EMi_Expansion}
\E{{M'_i}^{-1}|Y_i}&=\of[M\epsilon^df_2(Y_i)+O\of{M\epsilon^{\gamma+d}C(Y_i)}]^{-1}\nonumber\\
&\qquad\qquad\qquad\times \of{1+O\of{\frac{1}{M\epsilon^df_2(Y_i)}}}\nonumber\\
&=\of{M\epsilon^df_2(Y_i)}^{-1}\of[1+O\of{\epsilon^{\gamma}}+O\of{\frac{1}{M\epsilon^d}}]
\end{align}

Thus, \eqref{E_ratio} can be simplified as

\begin{align}\label{simple_exp}
\E{\frac{N'_1}{M'_1}\middle\vert Y_1}&=\frac{f_1(Y_1)}{\eta f_2(Y_1)}+O\of{\epsilon^\gamma}+O\of{\frac{1}{M\epsilon^d}}.
\end{align}

We use (Lemma 3.2 in \cite{Noshad2017}) and Remark \ref{variance_others}, and obtain
 \begin{align}\label{Expec_ratio}
\E{ \widetilde{g}\of{\frac{\eta N'_1}{M'_1}}\middle\vert Y_1} &= g\of{\frac{f_1(Y_1)}{f_2(Y_1)}}+O\of{\epsilon^\gamma}\nonumber\\
&\qquad+O\of{\frac{1}{M\epsilon^d}}+O(N^{-\half}).
\end{align}

Finally from \eqref{B1_7th} we get

\begin{align}\label{final_bias}
\mathbb{B}\of[\widehat{D}_{g}(X,Y)]&= O\of{\epsilon^\gamma}+O\of{\frac{1}{M\epsilon^d}}+O(N^{-\half})+ O(\frac{L}{F})\nonumber\\
&= O\of{\epsilon^\gamma}+O\of{\frac{1}{N\epsilon^d}},
\end{align}
where in the second equation we have used the upper bound on $L$ in Lemma \ref{Lemma_bound_L} and the fact that $M/N=\eta$.
Finally note that we can use a similar method with the same steps to prove the convergence of an estimator for R\'{e}nyi divergence.

\section{Discussion and Experiments}

In this section we compare and contrast the advantages of the proposed estimator with competing estimators, and provide numerical results. These show the efficiency of our estimator in terms of MSE rate and computational complexity.

\begin{table*}[ht]
\caption{Comparison of proposed estimator to Ensemble NNR \cite{Noshad2017}, Ensemble KDE \cite{Kevin16} and Mirror KDE \cite{Poczos2014_1}}
\centering
\begin{tabular}{c c c c c}
\hline\hline
Estimator &HB & NNR &Ensemble KDE  & Mirror KDE\\ [0.5ex] % inserts table %heading
\hline
 MSE Rate & $O(1/N)$ &$O(1/N)$ &$O(1/N)$ &$O(1/N)$ \\
 Computational Complexity &$O(N)$ &$O(kN\log N)$ & $O(N^2)$ & $O(N^2)$ \\
 Required Smoothness ($\gamma$) & $d/2$ &$d$ &$(d+1)/2$ & $d/2$ \\
 Extra Smooth Boundaries & No & Yes & Yes & Yes \\
 Online Estimation & Yes & No & No & No \\
 Knowledge about Boundary & No & No & No & Yes \\[1ex]

\hline
\end{tabular}
\label{table_compare}
\end{table*}

Table \ref{table_compare} summarizes the differences between the proposed optimum estimator (EHB) with other competing estimators: Ensemble NNR \cite{Noshad2017}, Ensemble KDE \cite{Kevin16} and Mirror KDE \cite{Poczos2014_1}. In terms of MSE rate, all of these estimators can achieve the optimal parametric MSE rate of $O(1/N)$. In terms of computational complexity, our estimator has the best runtime compared to others. The smoothness parameter required for the optimum MSE rate is stated in terms of number of required derivatives of the density functions. The proposed estimator is the first divergence estimator that requires no extra smoothness at the boundaries. It is also the first divergence estimator that is directly applicable to online settings, retaining both the accuracy and linear total runtime. Finally, similar to NNR and Ensemble KDE estimators, the proposed estimator does not require any prior knowledge of the support of the densities.

We next compare the empirical performance of EHB to NNR, and the Ensemble KDE estimators. The experiments are done for two different types of f-divergence; KL-divergence and $\alpha$-divergence defined in \cite{alpha}.
Assume that $X$ and $Y$ are i.i.d. samples from independent truncated Gaussian densities.
%%%% Figure 1 & 2%%%%
Figure \ref{figure1}, shows the MSE estimation rate of $\alpha$-divergence with $\alpha=0.5$ of two Gaussian densities with the respective expectations of $[0, 0]$ and $[0, 1]$, and equal variances of $\sigma^2=I_2$ for different numbers of samples. For each sample size we repeat the experiment $50$ times, and compute the MSE of each estimator. While all of the estimators have the same asymptotic MSE rate, in practice the proposed estimator performs better. The runtime of this experiment is shown in Figure \ref{figure2}. The runtime experiment confirms the advantage of the EHB estimator compared to the previous estimators, in terms of computational complexity.

%%%%  Figure %%%%%

\begin{figure}	
	\centering
	\includegraphics[width=1\columnwidth]{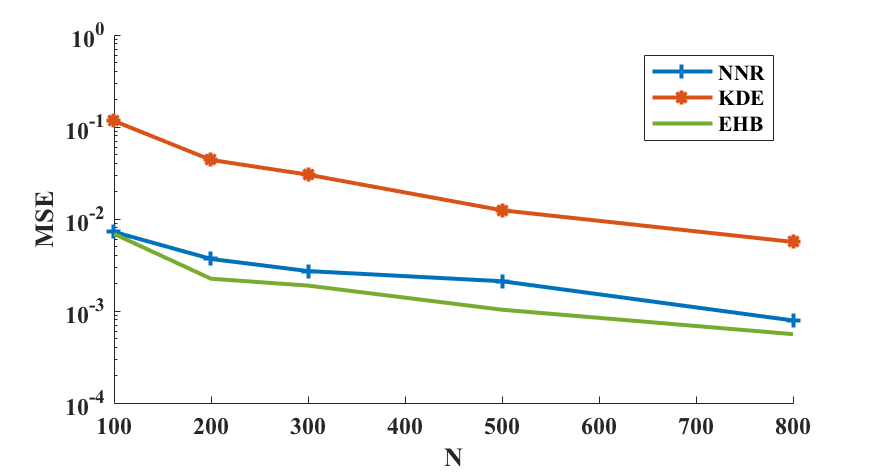}
    \caption{MSE comparison of $\alpha$-divergence estimators with $\alpha=0.5$ between two independent truncated 2D Gaussian densities with the respective expectations of $[0, 0]$ and $[0, 1]$, and equal variances of $\sigma_1^2=\sigma_2^2=I_2$, versus different number of samples.}
	\label{figure1}
\end{figure}

%%%%  Figure 2%%%%%
\begin{figure}	
	\centering
	\includegraphics[width=1\columnwidth]{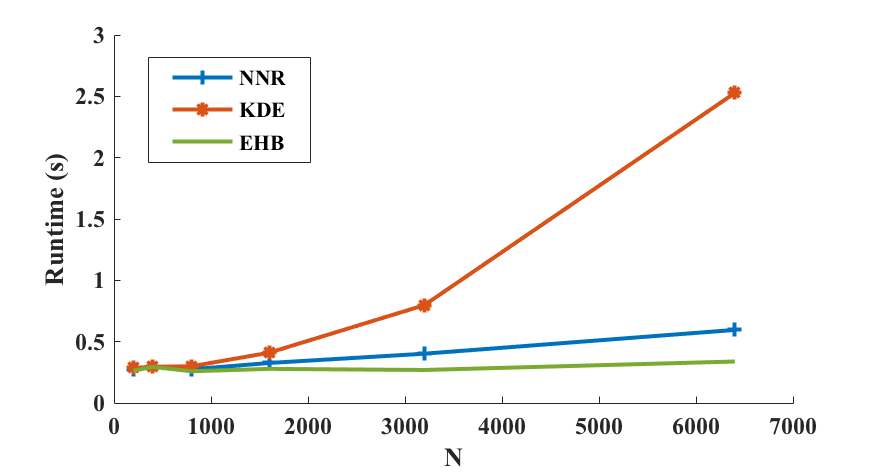}
    \caption{Runtime comparison of $\alpha$-divergence with $\alpha=0.5$ between two independent truncated 2D Gaussian densities with the respective expectations of $[0, 0]$ and $[0, 1]$, and equal variances of $\sigma_1^2=\sigma_2^2=I_2$, versus different number of samples.}
	\label{figure2}
\end{figure}

Figure \ref{figure2.5}, shows the comparison of the estimators of KL-divergence between two truncated Gaussian densities with the respective expectations of $[0, 0]$ and $[0, 1]$, and equal covariance matrices of $\sigma_1^2=\sigma_2^2=I_2$, in terms of their mean value and $\%95$ confidence band. The confidence band gets narrower for greater values of $N$, and EHB estimator has the narrowest confidence band.

%%%%  Figure 2.5 %%%%%
\begin{figure}	
	\centering
	\includegraphics[width=1\columnwidth]{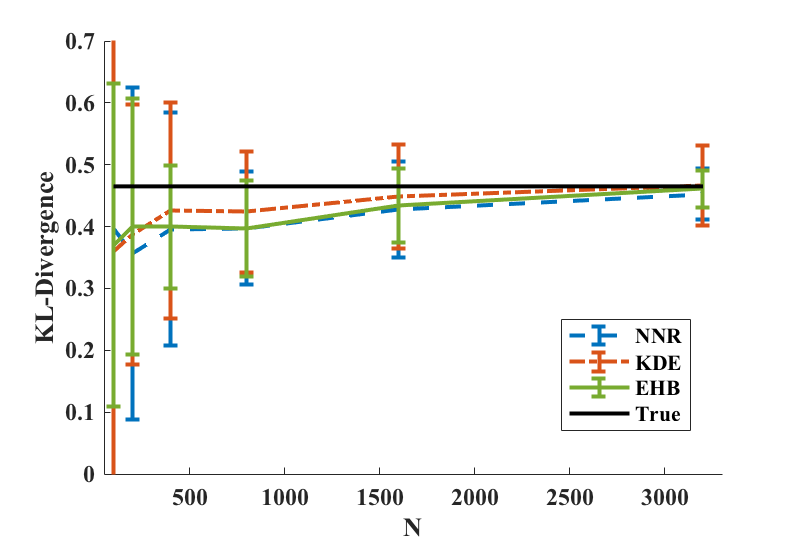}
    \caption{Comparison of the estimators of KL-divergence between two truncated Gaussian densities with the respective expectations of $[0, 0]$ and $[0, 1]$, and equal covariance matrices of $\sigma_1^2=\sigma_2^2=I_2$, in terms of their mean value and $\%95$ confidence band.}
	\label{figure2.5}
\end{figure}

%%%% Figure 3 %%%%%
In Figure \ref{figure3} the MSE rates of the three $\alpha$-divergence estimators are compared in dimension $d=4$, $\alpha=2$, for two independent truncated Gaussian densities with the expectations $\mu_1=\mu_2$ and covariances $\sigma_1^2=\sigma_2^2=I_4$, versus different number of samples.
\begin{figure}	
	\centering
	\includegraphics[width=1\columnwidth]{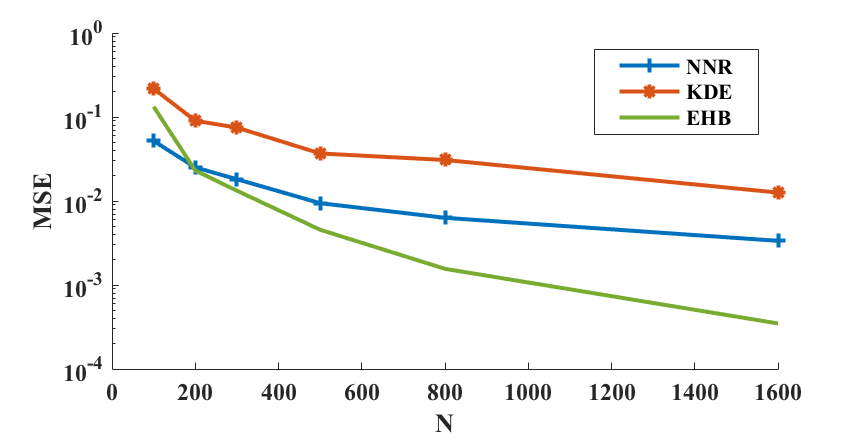}
    \caption{MSE estimation rate of $\alpha$-divergence with $\alpha=2$ between two independent truncated Gaussian densities with dimension $d=4$ and equal expectations $\mu_1=\mu_2$ and covariance matrices $\sigma_1^2=\sigma_2^2=I_4$, versus different number of samples.}
	\label{figure3}
\end{figure}
\section{Conclusion}
In this paper we proposed a fast hash based estimation method for f-divergence. We obtained bias and variance convergence rates, and validated our results by numerical experiments. Extending the method to hash-based mutual information estimation is a worthwhile topic for future work.

%\section{Acknowledgment}

%This research was partially supported by ARO grant W911NF-15-1-0479.
\bibliographystyle{ieeetr}
\bibliography{citations.bib}

\begin{thebibliography}{10}

\bibitem{cover2012}
T.~M. Cover and J.~A. Thomas, {\em Elements of information theory}.
\newblock John Wiley \& Sons, 2012.

\bibitem{moon2014}
K.~R. Moon and A.~O. Hero, ``Ensemble estimation of multivariate
  f-divergence,'' in {\em Information Theory (ISIT), 2014 IEEE International
  Symposium on}, pp.~356--360, IEEE, 2014.

\bibitem{structure2016}
K.~R. Moon, M.~Noshad, S.~Y. Sekeh, and A.~O. Hero~III, ``Information theoretic
  structure learning with confidence,'' in {\em Proc IEEE Int Conf Acoust
  Speech Signal Process}, 2017.

\bibitem{rrnyi1961}
A.~Rényi, ``On measures of entropy and information,'' in {\em Proceedings of
  the Fourth Berkeley Symposium on Mathematical Statistics and Probability,
  Volume 1}, pp.~547--561, University of California Press, 1961.

\bibitem{ali}
S.~M. Ali and S.~D. Silvey, ``A general class of coefficients of divergence of
  one distribution from another,'' {\em Journal of the Royal Statistical
  Society. Series B (Methodological)}, pp.~131--142, 1966.

\bibitem{poczos2011}
B.~P{\'o}czos and J.~G. Schneider, ``On the estimation of alpha-divergences.,''
  in {\em AISTATS}, pp.~609--617, 2011.

\bibitem{Poczos2014_2}
S.~Singh and B.~P{\'o}czos, ``Exponential concentration of a density functional
  estimator,'' in {\em Advances in Neural Information Processing Systems},
  pp.~3032--3040, 2014.

\bibitem{wang2009}
Q.~Wang, S.~R. Kulkarni, and S.~Verd{\'u}, ``Divergence estimation for
  multidimensional densities via-nearest-neighbor distances,'' {\em IEEE
  Transactions on Information Theory}, vol.~55, no.~5, pp.~2392--2405, 2009.

\bibitem{Noshad2017}
M.~Noshad, K.~R. Moon, S.~Y. Sekeh, and A.~O. Hero~III, ``Direct estimation of
  information divergence using nearest neighbor ratios,'' {\em arXiv preprint
  arXiv:1702.05222}, 2017.

\bibitem{kandasamy}
K.~Kandasamy, A.~Krishnamurthy, B.~Poczos, L.~Wasserman, {\em et~al.},
  ``Nonparametric {Von Mises} estimators for entropies, divergences and mutual
  informations,'' in {\em NIPS}, pp.~397--405, 2015.

\bibitem{Kevin16}
K.~R. Moon, K.~Sricharan, K.~Greenewald, and A.~O. Hero, ``Improving
  convergence of divergence functional ensemble estimators,'' in {\em IEEE
  International Symposium Inf Theory}, pp.~1133--1137, IEEE, 2016.

\bibitem{hash_KNN_graph}
Y.-m. Zhang, K.~Huang, G.~Geng, and C.-l. Liu, ``Fast {kNN} graph construction
  with locality sensitive hashing,'' in {\em Joint European Conference on
  Machine Learning and Knowledge Discovery in Databases}, pp.~660--674,
  Springer, 2013.

\bibitem{LSH_KNN}
Q.~Lv, W.~Josephson, Z.~Wang, M.~Charikar, and K.~Li, ``Multi-probe {LSH}:
  efficient indexing for high-dimensional similarity search,'' in {\em
  Proceedings of the 33rd international conference on Very large data bases},
  pp.~950--961, VLDB Endowment, 2007.

\bibitem{znidaric2005}
M.~Znidaric, ``Asymptotic expansion for inverse moments of binomial and
  {Poisson} distributions,'' {\em arXiv preprint math/0511226}, 2005.

\bibitem{Poczos2014_1}
S.~Singh and B.~P{\'o}czos, ``Generalized exponential concentration inequality
  for {Renyi} divergence estimation.,'' in {\em ICML}, pp.~333--341, 2014.

\bibitem{alpha}
A.~Cichocki, H.~Lee, Y.-D. Kim, and S.~Choi, ``Non-negative matrix
  factorization with $\alpha$-divergence,'' {\em Pattern Recognition Letters},
  vol.~29, no.~9, pp.~1433--1440, 2008.

\end{thebibliography}
\onecolumn
\newpage

%%%%% VARIANCE PROOF %%%%%
\section*{A. Variance Proof}
\begin{proof}[\textbf{Proof of Theorem \ref{variance}}: ]
 The proof is based on Efron-Stein inequality. We follow similar steps used to prove the variance of NNR estimator in \cite{Noshad2017}. Note that the proof for variance of 
 $\rho_i=N_i/(M_i)$ is contained in the the variance proof for $\widehat{D}_g(X,Y)$.
Assume that we have two sets of nodes $X_i$, $1\leq i \leq N$ and $Y_j$ for $1\leq j \leq M$. Here for simplicity we assume that $N = M$, however, the extension of the proof to the case when $M$ and $N$ are not equal, is straightforward, by considering a number of virtual points, as considered in \cite{Noshad2017}. Define $Z_i:=(X_i,Y_i)$. For using the Efron\hyp Stein inequality on $Z:=(Z_1,...,Z_N)$, we consider another independent copy of $Z$ as  $Z':=(Z'_1,...,Z'_N)$ and define $Z^{(i)}:=(Z_1,...,Z_{i-1},Z'_i,Z_{i+1},...,Z_N)$. Define $\widehat{D}_g(Z):=\widehat{D}_g(X,Y)$. By applying Efron\hyp Stein inequality we have

\begin{align}\label{variance_main}
\mathbb{V}\of[\widehat{D}_g(Z)] &\leq \frac{1}{2} \sum_{i=1}^N \mathbb{E}\left[\left(\widehat{D}_g(Z)-\widehat{D}_g(Z^{(i)})\right)^2\right]\nonumber\\
%%% First
&= \frac{N}{2} \mathbb{E}\left[\left(\widehat{D}_g(Z)-\widehat{D}_g(Z^{(1)})\right)^2\right]  \nonumber\\
%%% Second Line %%%
&\leq \frac{N}{2}\E{\left(\frac{1}{N}\sum_{\substack{i\leq F \\ M_i>0}} M_i\widetilde{g}\of{\frac{\eta N_i}{M_i}}- \frac{1}{N}\sum_{\substack{i\leq F \\ M_i>0}} M_i^{(1)}\widetilde{g}\of{\frac{\eta N_i^{(1)}}{M_i^{(1)}}} \right)^2} \nonumber\\
%%% Third Line %%%
&=\frac{1}{2N}\E{\left(\sum_{\substack{i\leq F \\ M_i>0}}\left(M_i\widetilde{g}\of{\frac{\eta N_i}{M_i}} - M_i^{(1)}\widetilde{g}\of{\frac{\eta N_i^{(1)}}{M_i^{(1)}}} \right)\right)^2} \nonumber\\
%%% Forth Line %%%
&=\frac{1}{2N}O\of{1}= O(\frac{1}{N}).
\end{align}
where in the last line we used the fact that $M_i$ and $M_i'$ can be different just for two of $i\leq F$, and that difference is just $O(1)$. So, the proof is complete.
\end{proof}

%%%%% End  Variance %%%%%

%%%%% Proof of Ensemble

\section*{B. Proof of Theorem \ref{ensemble_theorem}}

Assume that the densities have bounded derivatives up to the order $q$. Then the Taylor expansion of $f(y)$ around $f(x)$ is as follows 
\begin{equation}\label{taylor}
f(y)=f(x)+\sum_{|i|\leq q}\frac{D^i f(x)}{i!}\|y-x\|^i+O\of{\|y-x\|^q}.
\end{equation}

Therefore, similar to \eqref{ENi_Expansion} and using \eqref{int_moment} we can write

\begin{align}\label{Ensemble_1}
\E{N'_i} &=N\int_{x\in B_i}f_1(x)dx\nonumber\\
&=N{\displaystyle\int_{x\in B_i}f(Y_i)+\sum_{|j|\leq q}\frac{D^j f(Y_i)}{j!}\|x-Y_i\|^j+O\of{\|x-Y_i\|^q}} dx\nonumber\\
&=N\epsilon^df_1(Y_i)+ N\sum_{|j|\leq q}\frac{D^j f(Y_i)}{j!}C_j(Y_i)\epsilon^{|j|+d}+O\of{NC_q(Y_i)\epsilon^{q+d}}\nonumber\\
&=N\epsilon^d\of[f_1(Y_i)+ \sum_{l=1}^q C'_l(Y_i)\epsilon^{l}+O\of{C_q(Y_i)\epsilon^{q}}],
\end{align}
where $$C'_{|j|}(Y_i):=\sum_{|j|\leq q}\frac{D^j f(Y_i)}{j!}C_j(Y_i).$$

Similarly we obtain

\begin{align}\label{Ensemble_2}
\E{(M'_i)^{-1}} &=M^{-1}\epsilon^{-d}\of[f_2(Y_i)+ \sum_{l=1}^q C'_l(Y_i)\epsilon^{l}+O\of{C_q(Y_i)\epsilon^{q}}]^{-1}\of{1+O\of{\frac{1}{M\epsilon^df_2(Y_i)}}}.
\end{align}

The rest of the proof follows by using the same steps as used in equations \eqref{simple_exp}-\eqref{final_bias}, and we get

\begin{align}\label{final_ensemble}
\mathbb{B}\of[\widehat{D}_{g}(X,Y)]&= \sum_{i=1}^q C''_i\epsilon^{i}+O\of{\frac{1}{N\epsilon^d}},
\end{align}
where $C''_1,...,C''_2$ are constants. 

Now are ready to apply the ensemble theorem (\cite{Kevin16}, Theorem 4). Let $\mathcal{T}:=\{t_1,...,t_T\}$ be a set of index values with $t_i<c$, where $c>0$ is a constant. Let $\epsilon(t):=\floor{tN^{-1/2d}}$. The proof completes by using the ensemble theorem in (\cite{Kevin16}, Theorem 4) with the parameters $\psi_i(t)=t^{i/d}$ and $\phi'_{i,d}(N)=\phi_{i,\kappa}(N)/N^{i/d}$. So the following weighted ensemble has the MSE convergence rate of $O(1/N)$:

\begin{align}\label{EHB_def}
\widehat{D}_w:=\sum_{t\in \mathcal{T}}w(t)\widehat{D}_{\epsilon(t)}.
\end{align}

%%%%%%%%%%%%%%%%% End Ensemble Theorem %%%%%%%

%%%%%%%%%%%%%. Online Estimation Theorem %%%
\subsection*{C. Proof of Theorem \ref{Online_Theorem}: }

We first argue that amortized runtime complexity of the online estimation algorithm is order $O(1)$ for each update after adding new samples. Note that when we add a new pair of samples $X_{N+1}$ and $Y_{N+1}$, if $N+1\neq 2^k$ for some integer $k$, we only find $H(X_{N+1})$ and $H(Y_{N+1})$ and update the corresponding $M_i$ and $N_i$, which take a constant time. But, only when $N+1=2^k$ for some integer $k$, we need $O(N)$ time complexity to update $\epsilon$ and therefore the hash function. Thus, if we have $N=2^k$ nodes added to the estimation algorithm, the total complexity due to rehashing, $T_H$, is as follows:

\begin{align}
T_H=1+2+2^2+...+2^k=2^{k+1}-1=2N-1.
\end{align}
So, the amortized runtime complexity per each time step is $O(\frac{2N-1}{N})=O(1)$. So overall, the amortized computational complexity is order $O(1)$.
Finally, note that since we update $\epsilon$ when $N$ doubles, it is at most by a factor of $2$ away from the optimum $\epsilon$. Since constant factor doesn't affect the asymptotic order of the bias error, the bias bound always holds for online estimation algorithm.

\end{document}